\documentclass{article}
\usepackage{spconf}
\usepackage{amsmath}
\usepackage{amsthm}
\usepackage{amsfonts}
\usepackage{amssymb}
\usepackage{mathtools}
\usepackage{graphicx}
\usepackage{color}
\usepackage{verbatim}
\usepackage{hyperref}
\usepackage{todonotes}
\usepackage{enumitem}
\usepackage{cite}
\newcommand{\white}{\color{white}}

\newcommand{\I}{\mathcal{I}}
\newcommand{\J}{\mathcal{J}}
\newcommand{\K}{\mathcal{K}}
\newcommand{\wb}{\bold{w}}
\newcommand{\Wb}{\bold{W}}
\newcommand{\xb}{\bold{x}}
\newcommand{\Xb}{\bold{X}}
\newcommand{\yb}{\bold{y}}
\newcommand{\Gb}{\bold{G}}
\newcommand{\Tb}{\bold{T}}
\newcommand{\D}{\mathcal{D}}

\newcommand{\Scal}{\mathcal{S}}
\newcommand{\St}{\bar{\mathcal{S}}}
\newcommand{\Z}{\mathbb{Z}}
\newcommand{\R}{\mathbb{R}}
\newcommand{\C}{\mathbb{C}}
\newcommand{\gb}{\bold{g}}
\newcommand{\Sb}{\bold{S}}
\newcommand{\Sbt}{\tilde{\bold{S}}}
\newcommand{\Sbh}{\hat{\bold{S}}}

\newcommand{\Bb}{\bold{B}}
\newcommand{\ab}{\bold{a}}

\newtheorem{Thm}{Theorem}[section]
\newtheorem{Cor}[Thm]{Corollary}

\newcommand{\ind}{\text{\color{white}.$\quad$}}

\title{Weighted Gradient Coding with Leverage Score Sampling}
\name{$\text{Neophytos Charalambides}^{\dagger}, \text{Mert Pilanci}^{\ddagger}, \text{ and Alfred O. Hero III}^{\dagger},$}
\address{$\text{\white .}^\dagger$EECS Department University of Michigan, $\text{\white .}^\ddagger$EE Department Stanford University}

\begin{document}
%\ninept
%
\maketitle
\begin{abstract}
A major hurdle in machine learning is scalability to massive datasets. Approaches to overcome this hurdle include compression of the data matrix and distributing the computations. \textit{Leverage score sampling} provides a compressed approximation of a data matrix using an importance weighted subset. \textit{Gradient coding} has been recently proposed in distributed optimization to compute the gradient using multiple unreliable worker nodes. By designing coding matrices, gradient coded computations can be made resilient to stragglers, which are nodes in a distributed network that degrade system performance. We present a novel \textit{weighted leverage score} approach, that achieves improved performance for distributed gradient coding by utilizing an importance sampling.
\end{abstract}

\begin{keywords}
Reed-Solomon codes, distributed gradient descent, linear regression, sketching, straggler mitigation.
\end{keywords}

% - - - - - - - - - - - - - - -
\section{Introduction}
\label{intro}

In modern day machine learning the \textit{curse of dimensionality} has been a major impediment to solving large scale optimization problems. Gradient methods are widely used in solving such problems, but computing the gradient is often cumbersome. A method for speeding up the gradient computation is by performing the necessary computations in a distributed manner, where a network of workers perform certain subtasks in parallel. A common issue which arises are stragglers: workers whose task may never be received, due to delay or outage. These straggler failures translate to erasures in coding theory. The authors of \cite{TLDK17} proposed \textit{gradient coding}, a scheme for exact recovery of the gradient when the objective loss function is additively separable. The exact recovery of the gradient is considered in several prior works, e.g., \cite{TLDK17,HASH17,OGU19,YA18,CMH20,RTTD17}, while gradient coding for approximate recovery of the gradient is studied in \cite{CP18,RTTD17,CPE17,WCP19,BWE19,WLS19,KKR19,HYKM19}. Gradient coding requires the central server to receive the subtasks of a fixed fraction of any of the workers. We obtain an extension based on balanced Reed-Solomon codes \cite{HLH16,HASH17}, introducing \textit{weighted gradient coding}, where the central server recovers a weighted sum of the partial gradients of the loss function.\\
$\ind$ Our extension also includes dimensionality reduction via \textit{sketching}, in which a matrix is sampled by row selection based on leverage scores. We adapt the leverage scores subsampling algorithm from \cite{MMY15,Wan15} to incorporate weights, which we refer to as \textit{weighted leverage score sampling}.\\
$\ind$ The proposed approach merges weighted gradient coding with weighted leverage score sampling, and significantly benefits from both techniques. The introduction of weights allows for further ``compression'' in our data matrix when the leverage scores are non-uniform. Our experiments show that the proposed approach succeeds in reducing the number of iterations of gradient descent; without significant sacrifice in performance.\\
$\ind$ The paper is organized as follows. In section \ref{str_problem_GC} we describe the ``straggler problem'' in gradient coding \cite{TLDK17}. In section \ref{WLSS} we introduce leverage score sampling, and modify existing dimensionality reduction techniques by introducing weights. In \ref{WGC} we present the weighted gradient coding scheme, which we combine with weighted leverage score sampling. In \ref{eq_gradients} we show equivalence of the gradients computed using our proposed scheme and those obtained by applying the leverage score sketching matrix \cite{Wan15}. Finally, in \ref{experiments} we present experimental results. The contributions of this paper are:
%$\ind$ The rest of the paper is organized as follows. In section \ref{str_problem_GC} we describe the ``straggler problem'' and gradient coding \cite{TLDK17}; which have recently attracted a lot of attention, while also setting up the necessary notation. In section \ref{WLSS} we introduce leverage score sampling as a technique for dimensionality reduction, and modify existing techniques by introducing weights; in order to get further reduction. In \ref{WGC} we give a gradient coding scheme which recovers the weighted sum of partial gradients, which we combine with the the scheme from \ref{WLSS} to benefit from both techniques in gradient methods. In \ref{eq_gradients} we show how under some mild assumptions, the gradient of our weighted scheme for the least squares loss function is equivalent to that obtained by applying the leverage score sketching matrix {\red --- this can be generalized to other loss functions}. Finally, in \ref{experiments} we present results for our scheme on binary logistic regression classification on the MNIST dataset, and linear regression on an artificial dataset, for which gradient descent was used.
%Our contributions in this work are:
\begin{itemize}[noitemsep,nolistsep]
  \item Introduction of a \textit{weighted} gradient coding scheme, that is robust to stragglers.
  %\item Propose a gradient coding scheme for recovering a weighted linear combination of the partial gradients, under the straggler problem scenario
  \item Incorporation of leverage score sampling into weighted gradient coding.
  %\item Combine it with leverage score sampling to benefit from both techniques (operations at workers drops by a factor of $\rho$), while also allowing for further compression when the scores are non-uniform
  \item Theory showing that perfect gradient recovery of \textit{leverage score sketching} occurs, under mild assumptions.
  \item Presentation of experiments that corroborate our theoretical results.
  %\item {\red Show through experiments that we can have gain in performance and number of iterations of gradient descent}
\end{itemize}

% - - - - - - - - - - - - - - -
\section{Stragglers and Gradient Coding}
\label{str_problem_GC}

\subsection{Straggler Problem}
\label{str_problem}

Consider a single central server node that has at its disposal a dataset $\D=\left\{(\xb_i,y_i)\right\}_{i=1}^N\subsetneq \R^p\times\R$ of $N$ samples, where $\xb_i$ represents the features and $y_i$ the label of the $i^{th}$ sample. The central server can distribute the dataset $\D$ among $n$ workers in order to solve the optimization problem
\begin{equation}
\label{th_star_pr}
  \theta^{\star} = \arg\min_{\theta\in\R^p}\left\{ \sum_{i=1}^N \ell(\xb_i,y_i;\theta) \right\}
\end{equation}
in an accelerated manner, where $L(\D;\theta)=\sum_{i=1}^N \ell(\xb_i,y_i;\theta)$ is a predetermined loss-function. The objective function in (\ref{th_star_pr}) can also include a regularizer $\lambda R(\theta)$ if necessary. A common approach to solving (\ref{th_star_pr}), is to employ gradient descent. Even if closed form solutions exist for (\ref{th_star_pr}), gradient descent can still be advantageous for large $N$.

The central server is capable of distributing the dataset appropriately, with a certain level of redundancy, in order to recover the gradient based on the full $\D$. As a first step we partition the $\D$ into $k$ disjoint parts $\{\D_j\}_{j=1}^k$ each of size $N/k$. The quantity
$$ g=\nabla_{\theta}L(\D;\theta)=\sum_{j=1}^k\nabla_{\theta}\ell(\D_j;\theta)=\sum_{j=1}^k g_j $$
is the gradient. We call the summands $g_j\coloneqq\nabla_{\theta}\ell(\D_j;\theta)$ \textit{partial gradients}.

In the distributed setting each worker node completes its task of sending back a linear combination of its assigned partial gradients. There can be different types of failures that can occur during the computation or the communication process. These failures are what we refer to as \textit{stragglers}, that are ignored by the main server: specifically, the server only receives $f\coloneqq n-s$ completed tasks ($s$ is the number of stragglers our scheme can tolerate). We denote by $\I\subsetneq[n]\coloneqq\{1,\cdots,n\}$ the index set of the $f$ fastest workers who complete their task in time. Once \textit{any} set of $f$ tasks is received, the central server should be able to decode the received encoded partial gradients, and therefore recover the full gradient $g$.

\subsection{Gradient Coding}
\label{GC}

Gradient coding is a procedure comprised of an encoding matrix $\Bb\in\C^{n\times k}$, and a decoding vector $\ab_{\I}\in\C^f$; determined by $\I$. Each row of $\Bb$ corresponds to an encoding vector with support $w$, and each column corresponds to a data partition $\D_j$. The columns are  with support $d$, so every partition is sent to $d$ workers. By $\Bb_{\I}\in\C^{f\times k}$ we denote the submatrix of $\Bb$ consisting only of the rows indexed by $\I$. The entry $\Bb_{ij}$ is nonzero if and only if, $\D_j$ was assigned to the $i^{th}$ worker. Furthermore, we have $n\cdot w=k\cdot d$; where $s=d-1$.

Each worker node is assigned a number of gradients from the partition, indexed by $\J_i\subsetneq[k]$. The worker is tasked to compute an encoded version of the partial gradients $g_j\in\R^p$ corresponding to its assignments. Denote by
$$ \gb \coloneqq {\begin{pmatrix} | & | & & | \\ g_1 & g_2 & \hdots & g_k \\ | & | & & | \end{pmatrix}}^T \in \R^{k\times p} $$
the matrix whose rows constitute the transposes of the partial gradients, and the received encoded gradients are the rows of $\Bb_{\I}\gb$. The full gradient of the objective (\ref{th_star_pr}) on $\D$ should be recoverable by applying $\ab_{\I}$
$$ g^T=\ab_{\I}^T(\Bb_{\I}\gb)=\bold{1}_{1\times k}\gb= \sum_{j=1}^kg_j^{T}. $$
Hence the encoding matrix $\Bb_{\I}$ must satisfy $\ab_{\I}^T\Bb_{\I}=\bold{1}_{1\times k}$ for \textit{all} of the ${{n}\choose{s}}$ possible index sets $\I$.

Every partition will be sent to an equal number of servers $d$, and each server will receive a total of $w$ distinct partitions. In section \ref{WGC}, we introduce a procedure for recovering a weighted sum of partial gradients, i.e. $\tilde{g}=\sum_{i=1}^k\wb_ig_i$. Once the gradient has been recovered, the central server can perform an iteration of its gradient based algorithm. 

% - - - - - - - - - - - - - - -
\section{Weighted Leverage Score Sampling}
\label{WLSS}

\subsection{Leverage Score Sketching}
\label{lvg_score_sk}

Consider a data matrix $\Xb\in\R^{N\times p}$ for $N>p$, which is to be compressed through an operation which preserves the objective function in (\ref{th_star_pr}). Consider also
$$ \theta^{\star}_{ols} = \arg\min_{\theta\in\R^p} \left\{ \|\Xb\theta-\yb\|_2^2 \right\}. $$
for $\Xb\in\R^{N\times p}$ and $\yb\in\R^N$. In the case of the over-constrained least squares problem; i.e. $N\gg p$, the dimension $N$ is to be reduced to $r>p$, which corresponds to the number of constraints. Define the \textit{sketching matrix} $\Sb\in\R^{r\times N}$ and consider the modified problem
\begin{equation}
\label{wls_formulation}
  \tilde{\theta}_{ols} = \arg\min_{\theta\in\R^p} \left\{ \|\Sb(\Xb\theta-\yb)\|_2^2 \right\}.
\end{equation}

The leverage scores $\{\ell_i\}_{i=1}^N$ of matrix $\Xb$ are defined as the diagonal entries of the projection $P_\Xb=\Xb\Xb^{\dagger}$. Specifically $\ell_i=(P_{\Xb})_{ii}$. An equivalent definition of $\{\ell_i\}_{i=1}^N$ is via the reduced left singular vectors matrix $U\in\R^{N\times p}$, yielding $\ell_i=\|U_{(i)}\|_2^2=\|U_{(i)}\|_F^2$, for $U_{(i)}$ the $i^{th}$ row of $U$. A distribution is defined over the rows of $\Xb$ by normalizing these scores, where each row of $\Xb$ has respective probability of being sampled $\pi_i=\ell_i/\sum_{j=1}^N\ell_j=\ell_i/\|U\|_F^2$. Other methods for approximating the leverage scores are available \cite{DMMW12,RCCR18,SS11}, that do not require directly computing the singular vectors. We also note that a multitude of other efficient constructions of sketching matrices and iterative sketching algorithms have been studied in the literature \cite{pilanci2016fast,PilWai14a,PilWai14b,pilanci2017newton}.
%\todo[inline]{We should cite references on fast estimation of leverage scores: \href{https://papers.nips.cc/paper/7810-on-fast-leverage-score-sampling-and-optimal-learning.pdf}{example1}}

The \textit{leverage score sketching matrix} \cite{MMY15} is comprised of two matrices, a sampling matrix $S_{\Xb}\in\{0,1\}^{N\times r}$ and a rescaling matrix $D\in\R^{r\times r}$. The sketching matrix $\Sbt$ is constructed in two steps:
\begin{enumerate}[noitemsep]
  \item randomly sample with replacement $r>p$ rows from $\Xb$, based on $\{\pi_i\}_{i=1}^N$
  \item rescale each sampled row by $\frac{1}{\sqrt{r\pi_i}}$
\end{enumerate}
for which $\Sbt=D\cdot S_{\Xb}^T$. The modified least squares problem is then solved to obtain an approximate solution $\tilde{\theta}_{ols}$.

\subsection{Weighted Leverage Score Sampling}
\label{WLSS_subsection}

The sampling procedure described in \ref{WLSS} reduces the size of the dataset. We define the \textit{compression factor} as $\rho\coloneqq\frac{N}{r}\in\Z_+$, i.e. $r|N$, and we require that $k|r$. Different from section \ref{GC} where $k$ partitions were considered, here and in the sequel we consider $K=\rho\cdot k$ equipotent partitions $\{\D_i\}_{i=1}^K$ of size $N/K$, out of which we will draw $k$ \textit{distinct} data partitions. The sampling distribution is defined next.

The leverage score of each sampled partition is defined as the sum of the normalized leverage scores of the samples comprising it. That is; $\Pi_i\coloneqq\sum_{j:x_j\in \D_i}\pi_j$. In a similar manner to the definition of the leverage scores in \ref{lvg_score_sk} $\Pi_i=\|U_{(\K_i)}\|_F^2/\|U\|_F^2$, where $U_{(\K_i)}$ denotes the submatrix consisting only of the rows corresponding to the elements in $\D_i$. The proposed compression matrix $\Sb_\text{p}$ is determined as follows:
\begin{enumerate}[noitemsep]
  \item randomly sample with replacement $k$ parts in the partition of $\{\D_i\}_{i=1}^K$, based on $\{\Pi_i\}_{i=1}^K$
  \item retain each sampled part in the partition only once, and count how many times each part was drawn --- these counts correspond to entries in the \textit{weight vector} $\wb$
  \item rescale each part in the partition by $\frac{1}{\sqrt{r\Pi_i}}$.
\end{enumerate}
For the sampling matrix first construct $S_\text{part}\in\{0,1\}^{K\times k}$, which has a single $1$ entry in every column, corresponding to the distinct data partitions drawn, and then assign $S_{\Xb_\text{p}} = S_\text{part} \otimes \bold{I}_{N/K} \in \{0,1\}^{N\times r}$. For the rescaling matrix define $D_\text{part}\in\R^{k\times k}$ with diagonal entries $1/\sqrt{r\Pi_j}$ based on the $k$ drawn parts of the partition, and then $D_\text{p} = D_\text{part} \otimes \bold{I}_{N/K} \in \R^{r\times r}$. The final compression matrix is $\Sb_\text{p} \coloneqq D_\text{p}\cdot S_{\Xb_\text{p}}^T$.

Note that $S_\text{part}$ has no repeated columns. Note also that in the proposed distributed computing framework, we do not directly multiply by $\Sb_{\text{p}}$, but simply retain the sampled parts of the partition and rescale them before distributing the data.

% - - - - - - - - - - - - - - -
\section{Weighted Gradient Coding}
\label{WGC}

The objective is to recover the weighted sum of the partial gradients $\tilde{g}$, as depicted in Figure 1, subject to at most $s$ erasures. This may be achieved by extending the construction proposed in \cite{HASH17}. The main idea in \cite{HASH17} is to use balanced Reed-Solomon codes \cite{HLH16}, which are evaluation polynomial codes. Each column of the encoding matrix $\Bb$ corresponds to a partition $\D_i$ and is associated with a polynomial that evaluates to zero at the respective workers who have not been assigned that partition part. For more details the reader is referred to \cite{HASH17}.

\begin{Thm}
  Let  $\Bb$ and $\ab_{\I}$ be an encoding matrix and decoding vector from \cite{HASH17}, satisfying $\ab_{\I}^T\Bb_{\I}=\bold{1}_{1\times k}$ for any $\I$. Let $\tilde{\Bb}\coloneqq\Bb\cdot\textit{diag}(\wb)$ for any $\wb\in\C^{1\times k}$. Then $\ab_{\I}^T\tilde{\Bb}_{\I}=\wb$.
\end{Thm}

\begin{proof}
The properties of balanced Reed-Solomon codes imply the decomposition $\Bb_{\I}=\Gb_{\I}\Tb$, where $\Gb_{\I}$ is a Vandermonde matrix over the subgroup $U_n=\{a\in\C:a^n=1\}$ of the circle group, and the entries of $\Tb$ corresponds to the coefficients of polynomials; constructed such that their constant term is $1$, i.e. $\Tb_{(1)}=\bold{1}_{1\times k}$. The vector $\ab_{\I}^T$ is the first row of $\Gb_{\I}^{-1}$, for which $\ab_{\I}^T\Gb_{\I}=\bold{e}_1^T$; the first standard basis vector.

The matrix $\tilde{\Tb}=\Tb\cdot\textit{diag}(\wb)$ is equal to $\Tb$ with its columns each scaled by the respective entry of $\wb$, thus $\tilde{\Tb}_{(1)}=\wb$. A direct consequence of this is that $\ab_{\I}^T\tilde{\Bb}_{\I}=\bold{e}_1^T\tilde{\Tb}=\tilde{\Tb}_{(1)}=\wb$, which completes the proof.
\end{proof}

\begin{figure}
  \centering
  \label{wted_recovery}
    \includegraphics[scale=.15]{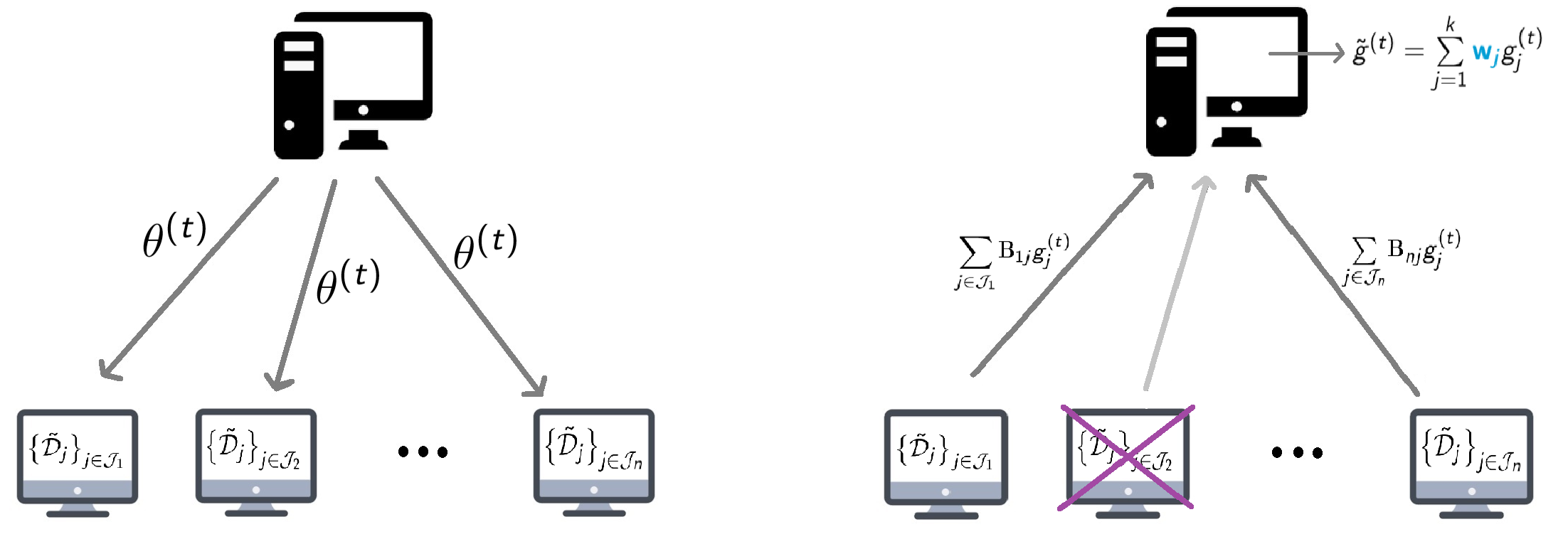}
    \caption{Schematic representation of communication, with the recovery of $\tilde{g}$ at iteration $t$ of gradient descent.}
\end{figure}

When combined with the leverage score sampling scheme, the expected time for computing the weighted gradient per iteration reduces by a factor of $\rho$.
%Assuming no communication delays for now, the expected time for computing the gradient per iteration is dropped by a factor of $\rho$ in terms of the dataset's size.

It is worth noting that weighted gradient coding has applications in other settings, e.g. if the partitions are drawn from noisy sources of different variance; one could select $\wb$ based on the estimated noise variances, obtaining improved gradient resiliency. This also directly relates to scenarios where heteroskedastic data is considered.
%It is also worth noting that weighted gradient coding can have applications in other settings, e.g. if the partitions are drawn from noisy sources of different variance; one could select $\wb$ based on the estimated noise variances for improved results.

% - - - - - - - - - - - - - - -
\section{Equivalence of Gradients}
\label{eq_gradients}

In this section we show that the proposed weighted scheme will satisfy the same properties as the matrix $\Sbt$ from section \ref{WLSS}, when gradient based methods are used to solve (\ref{wls_formulation}). This is a consequence of theorem \ref{eq_grad_thm}, applied to the least squares problem. The main reason for this equivalence property is that the weighted gradient $\tilde{g}$ matches the gradient when the leverage score sketching matrix is applied.

The pre-processing in section \ref{WLSS_subsection} which takes place on the data matrix can be accomplished by using another compression matrix $\Sbh$. This matrix is defined as
$$ \Sbh\coloneqq\sqrt{\Wb}\cdot\left(D_{\text{p}}\cdot S_{\Xb_{\text{p}}}^T\right) = \sqrt{\Wb}\cdot\Sb_{\text{p}} \in \R^{r\times N} $$
for $\Wb=\textit{diag}(\wb)\otimes \bold{I}_{N/K}\in\R^{r\times r}$. For the objective function of the optimization problem (\ref{wls_formulation})
\begin{equation}
\label{loss_sketched}
  L_{\text{S}}(\Sb,\Xb,\yb;\theta) \coloneqq  \sum\limits_{i=1}^r\big((\Sb\Xb\theta)_i-(\Sb\yb)_i\big)^2
\end{equation}
%$$ L_{\text{S}}(\Sb,\Xb,\yb;\theta) \coloneqq \|\Sb(\Xb\theta-\yb)\|_2^2 = \sum\limits_{i=1}^r\big((\Sb\Xb\theta)_i-(\Sb\yb)_i\big)^2 $$
we have the following.

\begin{Thm}
\label{eq_grad_thm}
  Let $\D_i=\{\xb_i\}$ for all $i\in[N]$ and $\sum_{i=1}^k\wb_i$ be the total number of random draws used to construct $\Sbt$ and $\Sbh$. For $L_{\mathrm{S}}$ as specified by (\ref{loss_sketched}):
  \begin{equation}
  \label{grad_equiv}
    \nabla_{\theta}L_{\mathrm{S}}(\Sbt,\D;\theta) = \nabla_{\theta}L_{\mathrm{S}}(\Sbh,\D;\theta)
  \end{equation}
  under any permutation of the rows of $\Sbt$ or $\Sbh$.
\end{Thm}

\begin{proof}
Denote the index list of sampled parts of the partition by $\Scal$, and their index set by $\St\subseteq[N]$, i.e. $\Scal$ has elements with multiplicity equal to $\wb_i$ and every element of $\St$ is distinct. By assumption $r=\sum_{i=1}^k\wb_i$, and $|\St|=k\leq|\Scal|=r$. Further note that $\Sbh^T\Sbh=\Sb_\text{p}^T\Wb\Sb_\text{p}$, and for the loss function (\ref{grad_equiv})

\begin{align*}
  \nabla_{\theta}L_{\text{S}}(\Sbt,\D;\theta) &= 2\Xb^T\left(\Sbt^T\Sbt\right)\left(\Xb\theta-\yb\right)\\
  &= 2\sum\limits_{l\in\Scal}\xb_l \cdot D_{ll}^2 \cdot  \left(\xb_l^T\theta-y_l\right)\\
  &= 2\sum\limits_{j\in\St}\wb_j\xb_j \cdot (D_{\text{p}})_{jj}^2 \cdot \left(\xb_j^T\theta-y_j\right)\\
  &= 2\sum\limits_{j\in\St}\xb_j \cdot \left(\sqrt{\wb_j}\cdot(D_{\text{p}})_{jj}\right)^2 \cdot \left(\xb_j^T\theta-y_j\right)\\
  &= 2\sum\limits_{j\in\St}\xb_j \cdot \left(\sqrt{\Wb}\cdot D_{\text{p}}\right)_{jj}^2 \cdot \left(\xb_j^T\theta-y_j\right)\\
  &= 2\Xb^T\left(\Sbh^T\Sbh\right)\left(\Xb\theta-\yb\right)\\
  &= \nabla_{\theta}L_{\text{S}}(\Sbh,\D;\theta)
\end{align*}

completing the proof.
\end{proof}

\begin{Cor}
\label{eq_grad_cor}
  If $\Sbt$ and $\Sbh$ were to be constructed by sampling partitions of more than one elements based on $\{\Pi_i\}_{i=1}^K$, conclusion (\ref{grad_equiv}) of theorem \ref{eq_grad_thm} remains valid.
\end{Cor}

The benefit of the proposed weighted procedure, is that the weights allow further compression of the data matrix $\Xb$; without affecting the recovery of the gradient. If $\{\pi_i\}_{i=1}^N$ and $\{\Pi_i\}_{i=1}^K$ are not close to being uniform, $\Sbh(\Xb-\yb)$ could have significantly fewer rows than $\Sbt(\Xb-\yb)$. Under the conditions of theorem \ref{eq_grad_thm}, the proposed matrix $\Sbh$ is applicable to the sketching procedures in \cite{MMY15,Mah16,Wan15,Woo14}.

A convergence result can also be established using \cite{Mah16} theorem 10. In particular, under appropriate assumptions, after $O(1/\varepsilon)$ iterations with termination criterion $\|\tilde{g}^{(t)}\|_2\leq\varepsilon$, the proposed weighted gradient coding procedure applied to (\ref{loss_sketched}) produces an $\varepsilon$-approximation to $\theta^{\star}_{ols}$.

%{\red From the point of view of adding vectors (geometrically), the partial gradients of the partitions sampled multiple times will be scaled accordingly. As a result, the gradient $g^{(t)}$ used at each iteration would be ``shifted'' more towards the direction of the partial gradients with higher weights. This is also the main idea behind the sketching technique, as the partitions sampled multiple times are of greater importance.}

% - - - - - - - - - - - - - - -
\section{Experiments}
\label{experiments}

\subsection{Binary Classification of MNIST}
\label{bin_cl}

For dataset partitions $k=20$, number of workers $n=50$, parts per worker $w=12$, and allocation of each part to distinct worker $d=30$, we trained a logistic regression model by applying gradient descent with the proposed method. The number of training samples was $N=10000$ and the procedure was tested on 1791 samples, for classifying images of four and nine from MNIST, of dimension $p=784$. The table below shows averaged results over six runs while varying $\rho$, when the weights were introduced and when they were not.
\begin{center}
\begin{tabular}{ll}
\begin{tabular}{ |p{.3cm}||p{.8cm}|p{.8cm}| }
%\label{log_regr_table_1}
\hline
\multicolumn{3}{|c|}{With weights} \\
\hline
$\rho$ & Error & Iter.\\
\hline
$4$ & $7.09\%$ & 25.5 \\
$10$ & $7.78\%$ & 14.67 \\
$20$ & $8.17\%$ & 12.5 \\
\hline
\end{tabular}
\hspace{5mm}
\begin{tabular}{ |p{.3cm}||p{.8cm}|p{.8cm}| }
\hline
\multicolumn{3}{|c|}{Without weights} \\
\hline
$\rho$ & Error & Iter.\\
\hline
$4$ & $9.08\%$ & 24.67 \\
$10$ & $8.32\%$ & 14.67 \\
$20$ & $8.5\%$ & 12.67 \\
\hline
\end{tabular}
\end{tabular}
\end{center}
Without any compression there was an average error of $4.37\%$. For $\rho=4,10$; $\wb$ becomes non-uniform and weighting results in better classification accuracy. As $\rho$ decreases, the distribution $\{\Pi_i\}_{i=1}^K$ becomes closer to uniform, leading to reduced advantage in using weighted gradient coding.

\subsection{Linear Regression}
\label{lr_random_data}

We retain the same setting of $k=20,n=50,w=12$ and $d=30$ from \ref{bin_cl}, and generate random data matrices with varying leverage scores as follows for the proposed weighted coding procedure. For all $i\in[20]$ we generate $\Xb_i\in\Z^{50\times 20}$ random matrices with entries from Uni$(-15i,15i)$, concatenate them and shuffle the rows to form $\Xb\in\Z^{1000\times 20}$. We then select an arbitrary $\yb\in\text{im}(\Xb)$, and add standard Gaussian noise $\vec{\epsilon}$. We ran experiments to compare the proposed weighted procedure for solving (\ref{wls_formulation}), for a fixed gradient descent step-size of $\alpha_t=10^{-7}$, with the same termination criterion $\|\nabla_{\theta}L_{\text{S}}(\Sbh,\D;\theta^{(t)})\|_2<0.1$ over 20 runs and error measure $\left\|\tilde{\theta}_{ols}-\Xb^{\dagger}\left(\yb+\vec{\epsilon}\right)\right\|_2^2$.
\begin{center}
%\begin{tabular}{ |p{.3cm}||p{1.52cm}|p{1.52cm}|p{1.2cm}| }
%\begin{tabular}{ |p{.3cm}||p{1.35cm}|p{1.35cm}|p{1.2cm}| }
\begin{tabular}{ |p{.4cm}||p{1.62cm}|p{1.62cm}|p{1.3cm}| }
%\label{log_regr_table_1}
\hline
%\multicolumn{4}{|c|}{Weighted vs. Unweighted} \\
\multicolumn{4}{|c|}{Average Number of Iterations and Error} \\
\hline
 $\rho$ & Weighted & Unweighted  & Error\\
% $\rho$ & Iter.-WTD & Iter.-UNW  & Error\\
\hline
$2$ & $107.55$ & $145.6$ & $O(10^{-5})$ \\
$4$ & $75.3$ & $84.4$ & $O(10^{-4})$ \\
\hline
\end{tabular}
\end{center}
As in \ref{bin_cl}, for lower $\rho$ the proposed weighted gradient coding approach achieves the same order of error as the unweighted, in fewer gradient descent iterations (on average).

\begin{figure}
  \centering
  \label{conv_wts}
    \includegraphics[scale=.08]{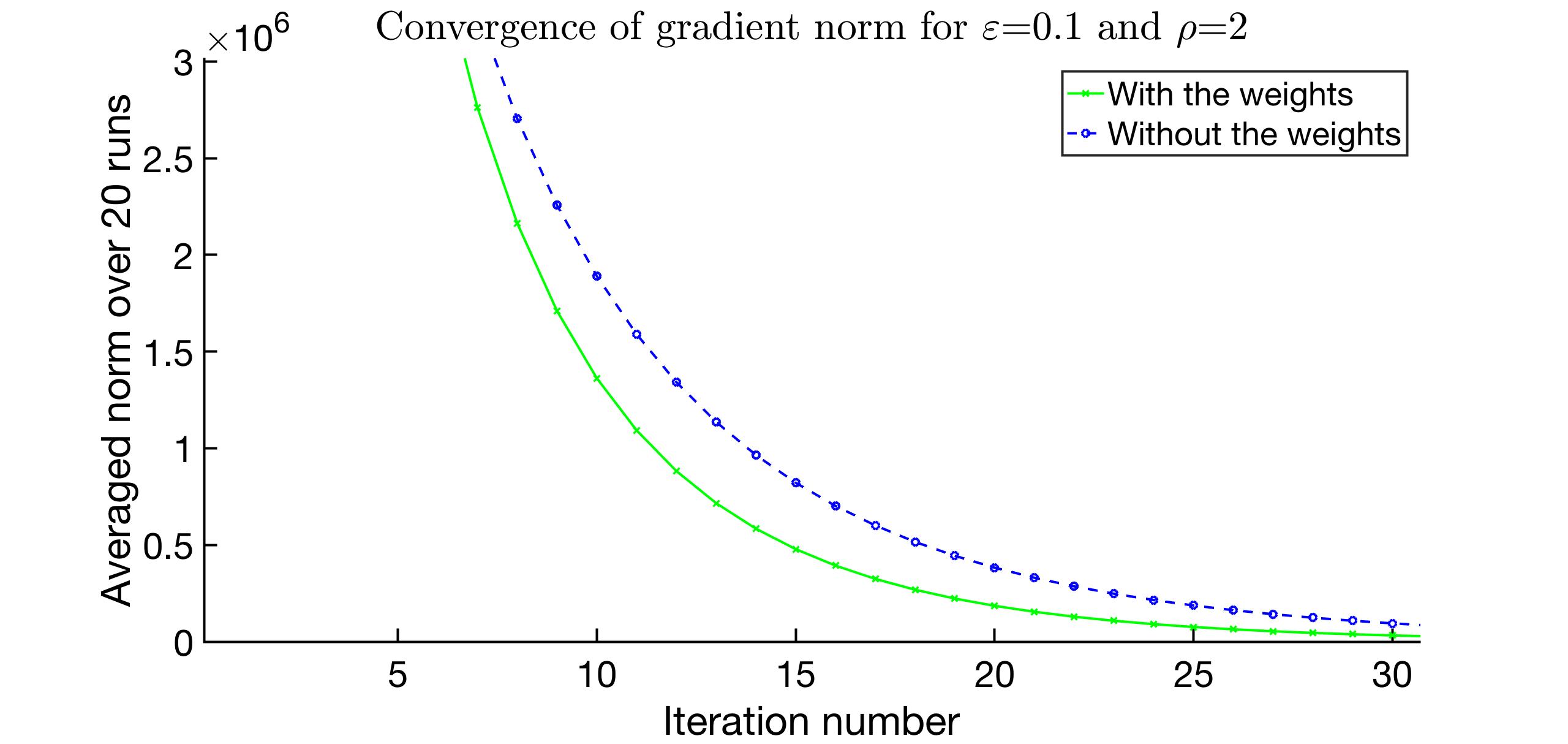}
    \caption{Convergence of the gradient norm in \ref{lr_random_data} for $\rho=2$, with and without the weights, where the norm at each iteration is an average of 20 different runs.}
\end{figure}

In Figure 2 we demonstrate the benefit of weighted versus unweighted. Even though the weights introduce a much higher gradient norm at first, it drops much faster and the termination criterion is met in fewer iterations.

% - - - - - - - - - - - - - - -
\section{Acknowledgement}
This work was partially supported by grant ARO W911NF-15-1-0479.

% - - - - - - - - - - - - - - -
%\section{\red Shifted material}

%The sampling is done by first partitioning $\D$ into $K=\rho\cdot k$ equipotent partitions of size $\frac{N}{K}$, i.e. $\D=\bigsqcup\limits_{j=1}^K\D_j$, and we draw (with replacement) until $k$ distinct partitions are drawn. Let us denote our new ``sub-dataset'' by $\Dt=\bigsqcup_{j=1}^k\Dt_j$, where by $\Dt_j$ we denote a reordering of the $k$ drawn partitions; and $|\Dt|=k\cdot\frac{N}{K}=r$. {\green We also keep track of how many times each partition is drawn; which we denote by {\red $\wb_j$}, and retain these sampled partitions, {\red which is what will be distributed to the workers}}. The total number of data samples we are now considering is $r$.

% - - - - - - - - - - - - - - -
%\section{REFERENCES}
%\label{sec:refs}

%\bibliographystyle{IEEEbib}
%\bibliography{strings,refs.bib}

%\bibliographystyle{plain}
\bibliographystyle{unsrt}
\bibliography{refs.bib}

\end{document}